\documentclass[11pt]{article}
\usepackage[margin=1.5in]{geometry}

\usepackage{amsfonts}
\usepackage{amssymb}
\usepackage{amsthm}
\usepackage{url}

\usepackage{hyperref}
\newtheorem{definition}{Definition}
\newtheorem{theorem}{Theorem}
\newtheorem{corollary}{Corollary}
\newtheorem{lemma}{Lemma}

\usepackage{booktabs} 
\usepackage[ruled]{algorithm2e} 

\SetAlFnt{\small}
\SetAlCapFnt{\small}
\SetAlCapNameFnt{\small}
\SetAlCapHSkip{0pt}
\IncMargin{-\parindent}

\usepackage{xcolor}
\usepackage{amsmath}
\usepackage{graphicx}
\newtheorem{fact}{Fact}

\title{Proof of Work With External Utilities}

\author{Yogev Bar-On \\ Tel Aviv University, Israel \\ \texttt{baronyogev@gmail.com} \and Ilan Komargodski \\ The Hebrew University of Jerusalem, Israel \\ \texttt{ilankom10@gmail.com} \and Omri Weinstein \\ The Hebrew University of Jerusalem, Israel \\ \texttt{omri.weins@gmail.com}}
\date{}
\begin{document}
\maketitle
\begin{abstract}
Proof-of-Work (PoW) consensus is traditionally analyzed under the assumption that all miners incur similar costs per unit of computational effort. In reality, costs vary due to factors such as regional electricity cost differences and access to specialized hardware. These variations in mining costs become even more pronounced in the emerging paradigm of \emph{Proof-of-Useful-Work} (PoUW), where miners can earn additional \emph{external} rewards by performing beneficial computations, such as Artificial Intelligence (AI) training and inference workloads \cite{schen2025proofs}.

Continuing the work of \cite{fiat2019energy}, who investigate equilibrium dynamics of PoW consensus under heterogeneous cost structures due to varying energy costs, we expand their model to also consider external rewards. We develop a theoretical framework to model miner behavior in such conditions and analyze the resulting equilibrium. Our findings suggest that in some cases, miners with access to external incentives will optimize profitability by concentrating their useful tasks in a single block. We also explore the implications of external rewards for decentralization, modeling it as the Shannon entropy of computational effort distribution among participants.

Empirical evidence supports many of our assumptions, indicating that AI training and inference workloads, when reused for consensus, can retain security comparable to Bitcoin while dramatically reducing computational costs and environmental waste.
\end{abstract}

\section{Introduction}
Proof-of-Work (PoW) has served as the cornerstone of permissionless blockchain consensus since Bitcoin's inception in 2008 \cite{nakamoto2008bitcoin}. Blockchains require consensus for all the participants to agree on the content of a public ledger. In PoW networks, the ledger is continually updated by the participants of the network (``miners''), via an operation called \emph{mining}. The mining operation requires solving a hard computational puzzle (finding a small hash), ensuring the security of the network due to the scarcity of computational resources. When applied to financial peer-to-peer payment transactions, the consensus prevents counterfeiting of the underlying currency.

As it turns out, modeling PoW-blockchain markets in a game-theoretic framework is subtle and may lead to very different market dynamics depending on how network interactions are modeled. For instance, in Nakamoto's original work \cite{nakamoto2008bitcoin} it was claimed and proved that as long as no miner has a majority of the mining power, no one would have reason to deviate from the protocol. This turned out to be false in a more careful modeling of a blockchain network (and also in real life) \cite{eyal2014majority}, since there are situations in which a miner is incentivized to delay the publication of a new block.

Another perspective was introduced by~\cite{fiat2019energy}. They investigated strategic deviations from the intended PoW protocol, due to varying energy costs and mining difficulty. Specifically, they showed that when energy costs are taken into account, counterintuitive and unintended strategic behavior occurs. For example, if miners have different energy efficiencies and are restricted to choose the same hash rate for many epochs, there is a unique pure equilibrium in which miners may choose to abstain from mining altogether if their efficiency is too low compared to others.

In recent years, an emerging paradigm in Proof-of-Stake consensus is \emph{restaking} \cite{buterin2021restaking, team2024eigenlayer} -- a mechanism that allows network participants to leverage staked assets across multiple applications concurrently, such as other blockchains. Essentially, the same resource is used to secure many protocols, thereby enhancing capital efficiency. However, the increased utility of staked assets through restaking introduces complexities in the protocol design, shifting participant incentives, and affecting security. Intuitively, the same amount of funds can be used for much more control, compared to the case where restaking is not available.

The closest analogue of restaking in Proof-of-Work blockchains is \textbf{Proof-of-Useful-Work} (PoUW) \cite{ball2017proofs,dong2019proofware, schen2025proofs}. In PoUW, mining involves computing a \emph{real-world} computational task (for which an external entity is willing to pay, such as AI model training, protein-folding, or shortest-paths in road networks), allowing miners to earn an external reward for performing the useful task in addition to the internal block reward. Similarly to restaking, the external task could simply be mining a different blockchain, a concept that is referred to as \emph{merge mining} \cite{judmayer2017merged}.

PoUW holds a great environmental advantage over traditional PoW, by reusing the same compute resources for different purposes. From this perspective, PoUW can be the ``best-of-both-worlds'' when comparing PoW to PoS. 

A key feature in Proof-of-Useful-Work networks is that mining requires both computational power \emph{and} the ability to receive external rewards on computation. This could be in the form of access to data needed for compute, clients that are willing to pay for it, or in the case of merge mining, the existence of other compatible blockchains. Hence, the underlying task required for mining is essentially ``restaked'' for different purposes.

\subsection{Our contributions}
In this work, we investigate the economics of varying costs for mining due to external rewards, using Proof-of-Useful-Work as a case study. Specifically, we put forth a general framework in which mining cost could depend on some external reward mechanism, which in turn implies a discount on their mining costs. In particular, our framework is general enough, so it can be used to capture other similar mechanisms like restaking.

Our work continues on the work of \cite{fiat2019energy}, where we observe that their analysis is limited to a setting where the energy cost scales linearly with the hash power. In contrast, in some settings, we see that super-linear cost functions might be a better approximation to reality (see Figure \ref{fig:quad_GPUs}).

We model the miners' cost function as a \emph{quadratic equation} parametrized with two variables, encompassing both the non-linear scaling of compute with the energy cost, and the effects of external rewards as cost subsidies. We characterize the Nash equilibrium implied by this model, and show that it holds a nice form with intuitive meaning (Corollary \ref{cor:linear-hashrate}). We also show that anyone with access to an external reward chooses to participate in consensus. Interestingly, we see that a dominating strategy is to concentrate useful work as much as possible, preferably to a single consensus epoch.

We then discuss the effects of PoUW on network security and decentralization, showing how evenly distributed access to useful work increases the security of the network rather than decreasing it (Lemma \ref{lem:decentralization}), compared to restaking in PoS networks, where usually the effects on security are negative.

\subsection{Related work}
Since the introduction of Bitcoin, there has been an extensive body of work trying to characterize the game-theoretic and economic equilibrium of PoW-based blockchains \cite{biais2019blockchain, budish2018economic,bonneau2015sok}.

The work of \cite{kroll2013economics} showed that infinitely many equilibria exist where Bitcoin participants behave inconsistently with the protocol. \cite{eyal2014majority,sapirshtein2016optimal,saad2019countering,negy2020selfish,bahrani2024undetectable} investigated the so called \emph{selfish mining} phenomena, where miners are incentivized to not reveal the solution for the computational puzzle.

Other works \cite{carlsten2016instability,chen2019axiomatic} researched the effects of mining rewards on the economic security of the network. Equilibria conditions for miners were investigated in \cite{toda2020game,chaidos2023blockchain}, where \cite{fiat2019energy,alberto2022evolutionary} considered varying mining costs like we do.

The security of restaking mechanisms was examined by \cite{chitra2024much,durvasula2024robust}. To quantify decentralization of blockchain networks, many works \cite{gupta2017gini,wu2019information,gochhayat2020measuring,lin2021measuring}, used different metrics such as Shannon entropy \cite{shannon1948mathematical}, Gini coefficient \cite{gini1936measure} and Nakamoto coefficient \cite{srinivasan2017quantifying}.

Suggestions of practical Proof-of-Useful-Work mechanisms were introduced by \cite{ball2017proofs,dong2019proofware}, with \cite{baldominos2019coin,lihu2020proof,chenli2020dlchain, schen2025proofs} specifically discussing using Artificial Intelligence (AI) workloads as a venue for external mining rewards. We are not aware of any prior work of analyzing the game-theoretic aspects of PoUW.

\subsection{Outline}
In the following section, we introduce our mathematical model and assumptions on compute costs and mining rewards under varying energy costs and access to external rewards (useful tasks). In Section \ref{sec:equilibria}, we then show how the compute power is distributed between miners in various cases under equilibrium, and the impact it has on their strategies. We continue to discuss the effects of varying cost functions to network decentralization in Section \ref{sec:decentralization}, and conclude in Section \ref{sec:conclusions}. 

\section{Model and Preliminaries}
We are interested in a system of $n$ miners spending compute resources in order to gain a reward, consisting of a \emph{block reward} and an \emph{external reward}. We denote by $h_i \geq 0$ the amount of compute resources the miner $i$ spends in a single time unit, and refer to it as the \emph{hash rate} (although it is important to note that the compute resources are not necessarily spent on hashing).

It is a known result \cite{chen2019axiomatic} that under some reasonable assumptions (such as symmetry and sybil-proofness), the only fair way to allocate a block reward is the \emph{proportional allocation rule}: 
\[
r_i\triangleq \frac{h_i}{H}R
\]
where $R>0$ is the block reward, $r_i$ is the \emph{expected} block reward of miner $i$, and $H=\sum_{i\in [n]} h_i$.

Let $H_{-i}=H-h=\sum_{j \in [n]\setminus \{i\}} h_j$ be the total hash rate of all miners except miner $i$. The net reward (utility) of miner $i$ as a function of their hash rate would then be:
\begin{equation}\label{eq:utility}
u_i(h) \triangleq r_i(h) - c_i(h) = \frac{h}{h + H_{-i}}R - c_i(h)
\end{equation}
where $c_i(h)$ is the cost of miner $i$ to produce a hash rate of $h$ (including the external reward, so the cost may be negative).

\subsection{Modeling the cost function}
The mining cost can be split into two parts: the actual compute cost, minus external rewards (which we call ``coupons''). To make our framework easily generalizable to more use cases, we treat the sum of both as a single cost function.

Our model attempts to capture a broad family of cost functions with enough structure to allow for an analytic solution. To this end, we identify two key properties for the cost function that should hold in practice:
\begin{itemize}
    \item It does not cost anything to not participate. In other words, $c_i(0) = 0$ for all $i\in [n]$.
    \item The cost function is convex.
    \begin{itemize}
        \item The intuition for convexity is that as long as there are useful tasks available, the cost is subsidized by an external reward (which we refer to as \emph{compute coupons}). Thus, the cost grows slowly (or even decreases) with the hash rate. Once there are no external rewards available (the ``coupons'' are used up), the cost starts to increase more rapidly.
        \item Moreover, the cost of renting computing devices (such as GPUs) increases convexly with the hash rate, as can be seen in Figure \ref{fig:quad_GPUs}.
    \end{itemize}
\end{itemize}

\begin{figure}[ht]
    \centering
    \includegraphics[width=0.95\textwidth]{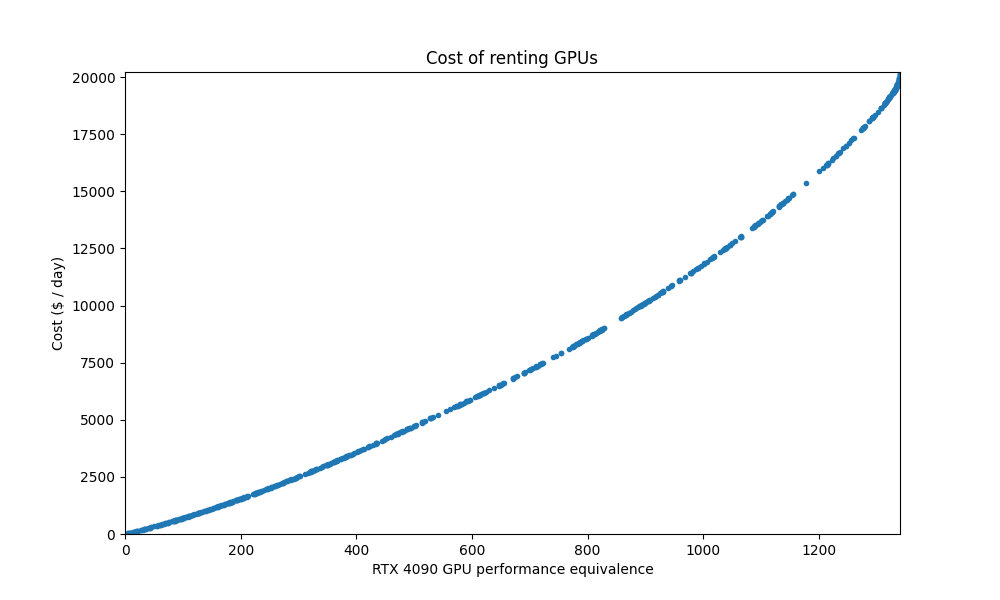}
    \caption{The cost per day of renting GPUs from \url{vast.ai}, with performance measured by the equivalent number of Nvidia RTX 4090 GPUs.}
    \label{fig:quad_GPUs}
\end{figure}

Our first observation is that if all the cost functions $c_i(h)$ are convex, we can see from Eq. (\ref{eq:utility}) that the utility functions $u_i(h)$ are concave. We can then use the following fact \cite{debreu1952social, glicksberg1952further, rosen1965existence}: 
\begin{fact}\label{fact:nash}
In a continuous concave game with a bounded strategy space, there exists a unique and pure Nash equilibrium.
\end{fact}
We may assume the hash rate is upper bounded (e.g., by the total kW electricity the world can produce), and thus Fact \ref{fact:nash} implies that an equilibrium exists for miners.

In practice, cost functions vary drastically and depend on many factors, such as energy cost, access to ASICs, availability of external rewards through useful tasks, and the type of tasks. Thus, we approximate the real cost function with a \emph{quadratic polynomial}, i.e., up to the second nonzero term in the Taylor approximation of the function. The choice for a second degree polynomial arises from the fact that in our context, the quadratic approximation has a ``semantic'' meaning, as the two terms roughly correspond to the cost of computation and the amount of external rewards. 

\begin{definition}[Quadratic Cost Function]
The quadratic cost function of miner $i$ is 
\begin{equation}\label{eq:cost}
c_i(h) \triangleq \frac{\alpha_i}{2}h^2 - \alpha_i\beta_i h
\end{equation}
where we refer to $\alpha_i > 0$ as the \emph{compute cost} and $\beta_i \geq 0$ as the compute \emph{coupons} of miner $i$.
\end{definition}
By taking the derivative of the cost function we get the marginal cost for one extra unit of hash rate:
\begin{equation}\label{eq:marginal}
c_i'(h) = \alpha_i (h - \beta_i).
\end{equation}
As we can see, the first $\beta_i$ units of hash rate (number of coupons) are ``free'', and the rest cost $\alpha_i h$, making the terminology appropriate: compute coupons can be thought of as the amount of external rewards available to the miner (e.g. due to having access to a client willing to pay for AI workloads). It is important to note that despite providing intuition for the quadratic approximation, the cost parameters do not directly map to real-world values.

\subsection{Decentralization and security}
We would like to quantify how external rewards affect the decentralization of the system, an important feature for permissionless blockchains. To do this, we follow a common approach to measure decentralization as the Shannon entropy of the miners' hash rate \cite{wu2019information}. The entropy of the hash rate measures how concentrated the compute power in the system is, with higher entropies meaning that it is more spread out between different miners.

\begin{definition}[Decentralization Coefficient]\label{def:entropy}
The \emph{decentralization coefficient} of a system with $n$ miners, each with hash rate $h_i$, is the Shannon entropy $S$ of the normalized hash rates:
\begin{equation}\label{eq:entropy}
D(h_1,...,h_n) \triangleq S\left(\frac{h_1}{H},...,\frac{h_n}{H}\right) = -\sum_{i\in [n]}{\frac{1}{H} h_i \log\left(\frac{h_i}{H}\right)},
\end{equation}
where $H=\sum_{i\in [n]}h_i$.
\end{definition}

We emphasize that the decentralization coefficient, as defined above, cannot be calculated in practice due to the inherent privacy in blockchain. Miners that appear to be different agents in the system may actually all be a single actor, artificially increasing the entropy without any increased security. For our intents and purposes, we only need to analyze how the decentralization coefficient is affected by PoUW, so we do not need to find the actual value for it.

\paragraph{\bf Rewards \rm}
One more thing to consider is that using entropy does not take into consideration the absolute value of the total hash rate in the system. Indeed, even with high entropy, if the total hash rate is low, a new actor could theoretically reduce the decentralization coefficient to a low value without many resources.

That said, we assume that in practice, the total hash rate correlates with the incentives of the miners. Specifically, we assume the block reward value $R$ is an increasing function of the total hash rate -- $R(H)$ increases with $H$.

To simplify our analysis, we will even assume for some of our results that the reward is a linear function of the total hash rate, i.e. $R(H)=\rho H$ where $\rho$ is the \emph{relative reward parameter}. We acknowledge this is a very simplified view, but is it useful for qualitative results, and is grounded in reality (see Figure \ref{fig:hashrates}).

\begin{figure}[ht]
    \centering
    \includegraphics[width=0.95\textwidth]{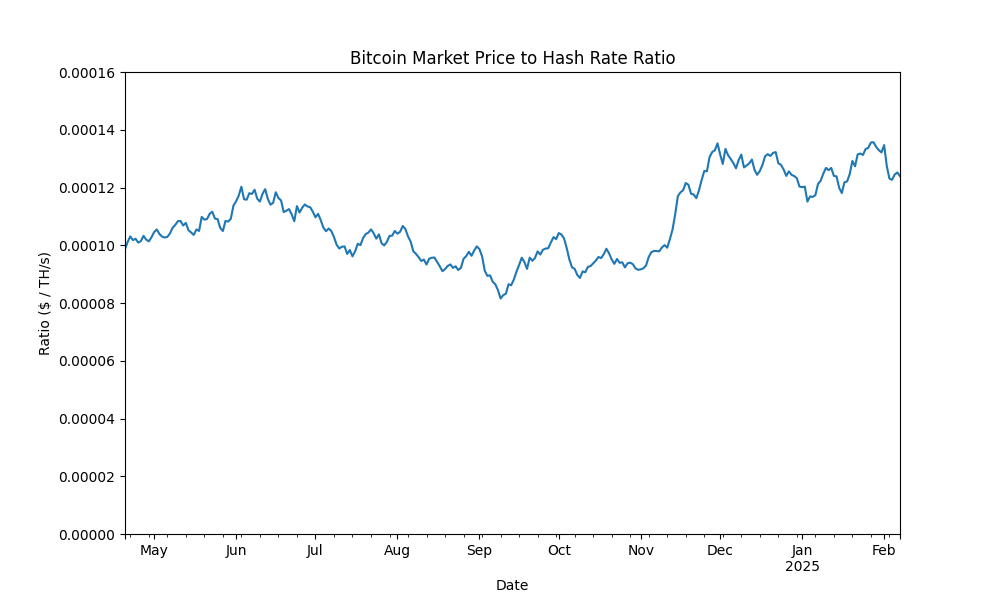}
    \caption{The ratio between the market price of Bitcoin and the total hash rate in the network (the relative reward parameter) since the last halving of the Bitcoin block reward. As you can see, the ratio is roughly stable. Source: \url{blockchain.com}.}
    \label{fig:hashrates}
\end{figure}

\section{Equilibrium of mining with external rewards}\label{sec:equilibria}
To compute the hash rates at equilibrium, we can compute the best-response strategy for each miner: the hash rate that will maximize their utility, assuming the other hash rates are known. We get the following result:
\begin{theorem}\label{thm:hashrate}
Given $n$ miners with quadratic cost functions, with compute costs $\{\alpha_i\}_{i\in [n]}$ and compute coupons $\{\beta_i\}_{i\in [n]}$, there is a unique pure Nash equilibrium such that for all $i\in [n]$ the hash rate satisfies:
\begin{equation}\label{eq:hashrate}
h_i = \frac{\alpha_i \beta_i H^2 + HR(H)}{\alpha_i H^2 + R(H) - HR'(H)},
\end{equation}
where $R(H)$ is the block reward in terms of $H=\sum_i h_i$.
\end{theorem}
\begin{proof}
As we previously discussed in Fact \ref{fact:nash}, the existence and uniqueness of a Nash equilibrium follow from the concavity of the utility function. To find the equilibrium, we will compute the best response of miner $i$. We will take the derivative of the utility function with respect to the hash rate, and compare it to $0$:
\begin{align*}
u'(h) &= r'_i(h) - c'_i(h) \\
&= \frac{H_{-i}R(h+H_{-i})}{(h + H_{-i})^2} + \frac{h}{h+H_{-i}}R'(h+H_{-i}) - \alpha_i (h - \beta_i) \\
&= \frac{(H-h)}{H^2}R(H) + \frac{h}{H}R'(H) - \alpha_i (h - \beta_i) = 0,
\end{align*}
And so:
\begin{align*}
&(H-h)R(H) + hHR'(H) - (h - \beta_i)\alpha_i H^2 = 0 \\
&(\alpha_i H^2 + R(H) - HR'(H)) h  = \alpha_i \beta_i H^2 + HR \\
&h = \frac{\alpha_i \beta_i H^2 + HR(H)}{\alpha_i H^2 + R(H) - HR'(H)}.
\end{align*}
Note that this solution is always non-negative and thus applicable.
\end{proof}

One perhaps surprising consequence of this result is that \emph{everyone} is participating under this framework. This becomes clear in Eq. (\ref{eq:marginal}), where we can see that the marginal cost near $0$ is small.

Now we can see an example of a linear reward function:
\begin{corollary}\label{cor:linear-hashrate}
If the block reward is linear in the total hash rate, $R(H)=\rho H$ for some relative reward parameter $\rho > 0$, we get that in equilibrium the hash rate of miner $i$ is:
\begin{equation}\label{eq:linear-hashrate}
  h_i = \beta_i + \frac{1}{\alpha_i} \rho.
\end{equation}
\end{corollary}
\begin{proof}
Substitute $R(H)=\rho H$ and $R'(H)=\rho$ in Eq. (\ref{eq:hashrate}):
\[
h_i = \frac{\alpha_i \beta_i H^2 + \rho H^2}{\alpha_i H^2} = \beta_i + \frac{1}{\alpha_i} \rho.
\]
\end{proof}

Corollary \ref{cor:linear-hashrate} presents an intuitive solution for equilibrium: first, use up all your coupons. Then, do the same work you would have done without any coupons available to you. Moreover, the amount of work beyond coupons has a simple form of linearly decreasing with the compute cost and linearly increasing with the relative reward.

Also worth noting is that \textbf{the hash rate at equilibrium does not depend on the actions of other miners}, and thus easily computable even in a private setting!

\subsection{Utility under linear rewards}
Now that we know the hash rate of each miner under equilibrium, we can also compute their net reward.
\begin{lemma}\label{lem:utility}
If the block reward is linear in the total hash rate, $R(H)=\rho H$ for some relative reward parameter $\rho > 0$, we get that in equilibrium the utility of miner $i$ is:
\begin{equation}\label{eq:linear-utility}
    u_i(h_i) = \frac{1}{2\alpha_i}\left(\alpha_i^2 \beta_i^2 + 2\alpha_i \beta_i \rho + \rho^2 \right).
\end{equation}
\end{lemma}
\begin{proof}
We start substituting Eq. (\ref{eq:cost}) in (\ref{eq:utility}):
\[
u_i(h) = \left(\frac{R(H)}{H} + \alpha_i \beta_i - \frac{\alpha_i}{2}h\right)h = \left(\rho + \alpha_i \beta_i - \frac{\alpha_i}{2}h\right)h.
\]

We can now use Eq. (\ref{eq:linear-hashrate}) and get:
\begin{align*}
u_i(h_i) &= \left(\rho + \alpha_i \beta_i - \frac{\alpha_i}{2}h_i\right)h_i \\
&= \left(\rho + \alpha_i \beta_i - \frac{\alpha_i}{2}\left(\beta_i + \frac{1}{\alpha_i} \rho\right)\right)\left(\beta_i + \frac{1}{\alpha_i} \rho\right) \\
&= \frac{1}{2\alpha_i}\left(\alpha_i^2 \beta_i^2 + 2\alpha_i \beta_i \rho + \rho^2 \right).
\end{align*}
\end{proof}

The specific equation is not significant in itself, and its structure mainly follows from our quadratic approximation of the cost function. The qualitative feature we do obtain from this result is that \emph{the utility is convex in the amount of coupons}.

This convexity means that \textbf{it is more profitable for a miner to ``save up'' coupons to use in a single block}, rather than distribute them equally. Say, for example, that miner $i$ has access to $n$ coupons that expire after two blocks. Set their compute cost to $\alpha_i=1$ for simplicity.

\begin{itemize}
    \item In the case the miner divides the coupons equally, such that in both blocks the amount of coupons is $\beta_i\frac{n}{2}$, we get the utility in both blocks is:
    \[
    u_i^{(1)} = u_i^{(2)} = \frac{1}{4}\left( \frac{n^2}{2} + 2\rho n + 2\rho^2 \right),
    \]
    making the total utility be:
    \[
    u_i = \frac{1}{2}\left( \frac{n^2}{2} + 2\rho n + 2\rho^2 \right).
    \]
    \item In the case the miner ``saves up'' the coupons, with $\beta_i=0$ in the first block and $\beta_i=n$ in the second, we get the utility in the first block is:
    \[
    u_i^{(1)} = \frac{1}{2} \rho^2,
    \]
    and in the second block it is:
    \[
    u_i^{(2)} = \frac{1}{2}\left( n^2 + 2\rho n + \rho^2 \right),
    \]
    making the total utility be:
    \[
    u_i = u_i^{(1)} + u_i^{(2)} = \frac{1}{2}\left( n^2 + 2\rho n + 2\rho^2 \right).
    \]
\end{itemize}

As we can see, using all the coupons at once increased the net reward of the miner across the two blocks by $\frac{n^2}{2}$.

\section{Effects of external rewards on decentralization}\label{sec:decentralization}

We can model Proof-of-Useful-Work mining as computing a real-world function $f(x)$ for some input $x$  \emph{chosen by the miner} itself for some external application. The base assumption is that there is some external entity willing to pay for the computation of $f(x)$. Hence, in contrast to traditional Proof-of-Work, miners require access to two types of resources: compute power (e.g., hardware and electricity) and \emph{useful tasks}, i.e., inputs $x$ for which the miner can obtain external rewards. We can therefore roughly translate the amount of useful tasks miner $i$ can access, to the number of \emph{coupons} $\beta_i$.

One concern regarding PoUW blockchains is that they are less decentralized due to mining being \emph{subsidized} by external actors. For example, if there is a single centralized actor with access to a large amount of coupons, they can maybe control the network for ``free''.

The key argument for alleviating this ``centralization'' concern is the realization that, in PoUW systems, coupons are a key resource required for \emph{efficient} mining, and in some cases (in contrast to energy), coupons are both \emph{scarce} and \emph{decentralized}. For example, we can see in Figure \ref{fig:ent_coupons} that AI workloads are distributed among many organizations around the world. Perhaps more importantly, coupons in many cases are \textbf{non-liquid} -- In a PoUW protocol where the computational overhead of participating in consensus is negligible, there is no incentive for rational, useful data owners to transfer their coupons at a low cost, as they are highly beneficial for them.

\begin{figure}[ht]
    \centering
    \includegraphics[width=0.95\textwidth]{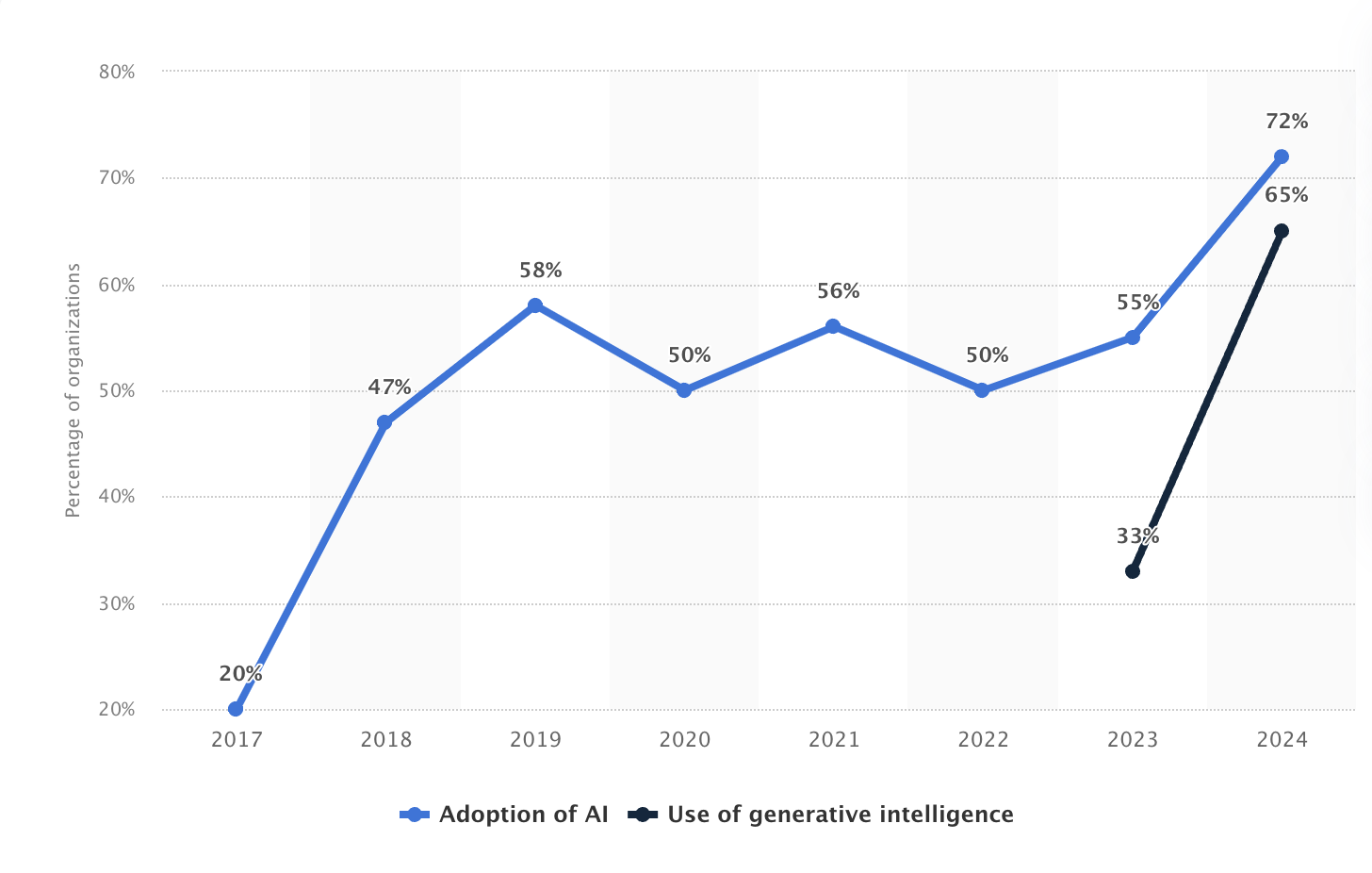}
    \caption{AI adoption rate among businesses worldwide overtime, as a proxy for AI \emph{data distribution}. Source: \url{statista.com}.}
    \label{fig:ent_coupons}
\end{figure}

From this perspective, PoUW can be viewed as a \emph{hybrid between Proof-of-Work and Proof-of-Stake}: mining requires both hoarding capital (coupons) and doing computation work. Hence, PoUW systems can benefit from the decentralization of both useful tasks and computing resources. Figure \ref{fig:ent_GPUs} illustrates this important distinction between the entropy of useful tasks availability versus the entropy of computing power used for mining.

\begin{figure}[ht]
    \centering
    \includegraphics[width=\textwidth]{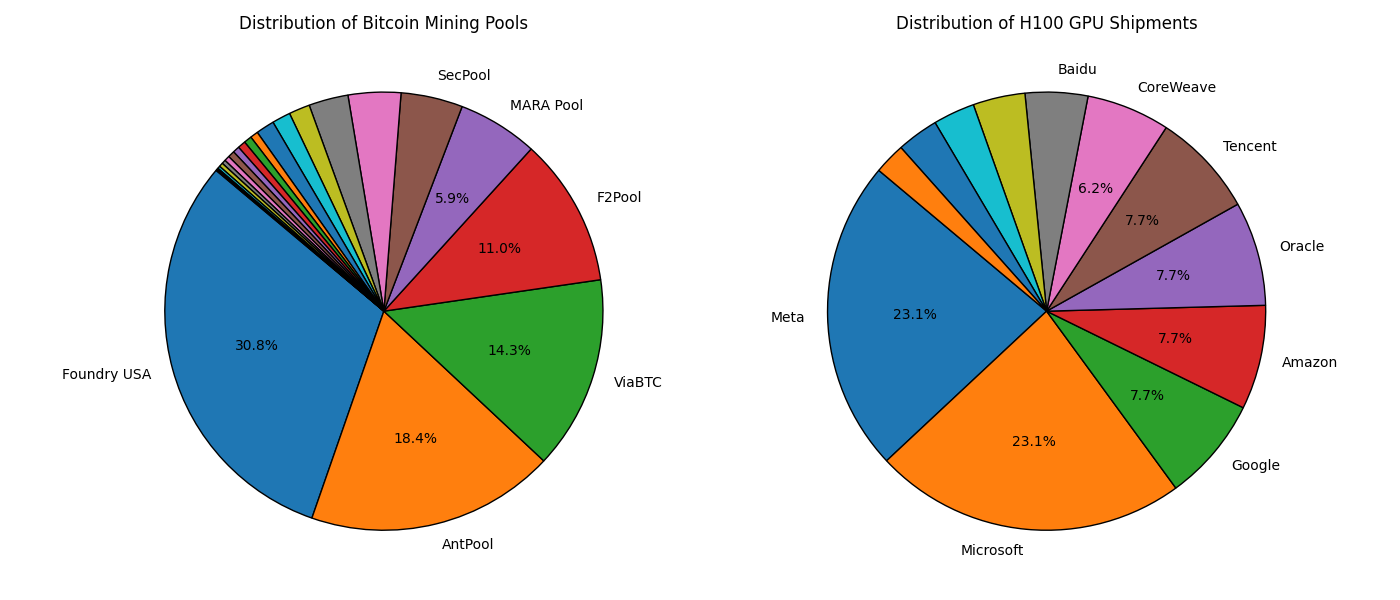}
    \caption{Distribution of hash rate of Bitcoin mining pools from Feb 2024 to January 2025 (left) vs. distribution of Nvidia H100 shipments to tech companies during 2023, as a proxy for their access to useful AI workloads (right). Sources: \url{https://hashrateindex.com} and \url{https://www.statista.com}.}
    \label{fig:ent_GPUs}
\end{figure}

\subsection{Properties of the decentralization coefficient under linear rewards}
Using Definition \ref{def:entropy}, we can quantify how decentralization is affected by coupons. For simplicity, we will analyze it under the assumption of linear rewards.

\begin{lemma}\label{lem:decentralization}
    Assuming a linear block reward $R(H)=\rho H$, let $h_1,...,h_n$ be the hash rates of all miners under equilibrium. Denoting $h^\alpha_i \triangleq \frac{1}{\alpha_i}\rho$ the hash rate due to the compute cost component, and $h^\beta_i \triangleq \beta_i$ the hash rate due to compute coupons, we get:
    \begin{equation*}
        D(h_1,....,h_n) \geq \frac{\sum_{i\in [n]}\beta_i}{H} D(h^\beta_1,...,h^\beta_n) + \frac{\sum_{i\in [n]}\frac{1}{\alpha_i}\rho}{H} D(h^\alpha_1,...,h^\alpha_n),
    \end{equation*}
    where $D$ is the decentralization coefficient.
\end{lemma}
\begin{proof}
Using Eq. (\ref{eq:linear-hashrate}) we gather that under equilbria, the hash rate of miner $i$ is $h_i = \beta_i + \frac{1}{\alpha_i}\rho$, and the normalized hash rate is $\hat{h}_i = \frac{h_i}{H} = \frac{\beta_i + \frac{1}{\alpha_i}\rho}{H}$. Let also $\hat{h}^\beta_i = \frac{\frac{1}{\alpha_i}\rho}{\sum_{j\in [n]}{\frac{1}{\alpha_j}\rho}}$ be the normalized hash rate of miner $i$ due to the compute cost component, and $\hat{h}^\beta_i = \frac{\beta_i}{\sum_{j\in [n]}{\beta_j}}$ the normalized hash rate due to the compute coupons component. We can now express the normalized hash as a weighted sum of its two normalized components:
\begin{align*}
\hat{h}_i &= \frac{h_i}{H}\\
&= \frac{\beta_i}{H} + \frac{\frac{1}{\alpha_i}\rho}{H}\\
&= \frac{\sum_{j\in [n]}{\beta_j}}{H} \frac{\beta_i}{\sum_{j\in [n]}{\beta_j}} + \frac{\sum_{j\in [n]}{\frac{1}{\alpha_j}\rho}}{H} \frac{\frac{1}{\alpha_i}\rho}{\sum_{j\in [n]}{\frac{1}{\alpha_j}\rho}}\\
&= \frac{\sum_{j\in [n]}{\beta_j}}{H} \hat{h}_i^\beta + \frac{\sum_{j\in [n]}{\frac{1}{\alpha_j}\rho}}{H} \hat{h}_i^\alpha.
\end{align*}

Setting $\lambda\triangleq\frac{\sum_{j\in [n]}{\beta_j}}{H} \in [0,1]$ as the relative magnitude of the compute coupons compared to the total hash rate, we get that:
\[
\hat{h}_i = \lambda \hat{h}_i^\beta + (1-\lambda) \hat{h}_i^\alpha.
\]

It is a well-known fact that the Shannon entropy function $S$ is concave \cite{cover1999elements}, and thus we can complete the proof:
\begin{align*}
 D(h_1,....,h_n) &= S(\hat{h}_1,...,\hat{h}_n)\\
 &\geq \lambda S(\hat{h}_1^\beta,...,\hat{h}_n^\beta) + (1-\lambda) S(\hat{h}_1^\alpha,...,\hat{h}_n^\alpha)\\
 &= \frac{\sum_{i\in [n]}\beta_i}{H} D(h^\beta_1,...,h^\beta_n) + \frac{\sum_{i\in [n]}\frac{1}{\alpha_i}\rho}{H} D(h^\alpha_1,...,h^\alpha_n).
\end{align*}
\end{proof}

Lemma \ref{lem:decentralization} states that the decentralization coefficient under equilibrium is \emph{at least} the weighted average of decentralization coefficients of the compute cost and compute coupons, weighted relative to their contribution to the total hash rate.

We can conclude the following result, matching our intuition:
\begin{itemize}
\item If the amount of total compute coupons available is low compared to the compute costs, the effects of coupons on decentralization are negligible.
\item Otherwise, if there is a significant amount of compute coupons, the entropy of their distribution across different actors will determine the decentralization of the network. Specifically:
\begin{itemize}
    \item If the entropy of coupons is low (e.g., one miner holds all the coupons), the decentralization coefficient might also be low, and the network might be less secure.
    \item If the entropy of coupons is high (e.g., divided equally between all the miners), the decentralization coefficient is high, increasing the security of the network.
\end{itemize}
\end{itemize}

\section{Conclusions}\label{sec:conclusions}

In this work, we introduced a general game-theoretic framework to analyze Proof-of-Work systems that incorporate external rewards, capturing the dynamics of Proof-of-Useful-Work (PoUW) mining. By approximating miner costs with a quadratic function, we highlighted key properties of the resulting equilibrium. First, we showed that these cost functions guarantee a unique pure Nash equilibrium in hash rate allocation: in equilibrium, every miner expends some non-zero effort, and the precise level of participation can be characterized solely from the miner’s cost and coupon parameters. Second, we demonstrated that miners optimally consume any available coupons in a concentrated manner (i.e., allocating them to a single block) to maximize profit. This result underscores an intrinsic incentive to bundle useful tasks, rather than spread them evenly over time.

Our findings further suggest that the overall impact of decentralization depends on the distribution of coupons (useful tasks). In scenarios where coupons are broadly available and non-transferable, PoUW encourages wider participation and contributes positively to the system’s decentralization. Conversely, if coupons remain concentrated in the hands of a few miners, their dominant presence could centralize hash power. We quantified this effect using the Shannon entropy of the hash rate distribution, providing a straightforward analytical tool to assess changes in decentralization when external rewards are introduced.

From a practical standpoint, our results support the potential of PoUW to retain PoW-level security with lower net costs and environmental effects, thanks to computational reuse. As real-world use cases for specialized tasks (e.g., AI training and inference) continue to grow, designing PoUW protocols that maintain robust levels of security and decentralization remains a crucial research direction. Future work might explore more intricate cost functions or time-varying rewards to build more precise models of miner behavior in multi-utility environments.

%

\bibliographystyle{plain}
\bibliography{refs}



\end{document}